\documentclass[onecolumn,amsmath]{revtex4}
\usepackage{physics}
\usepackage{graphicx}
\usepackage{amsthm}
\usepackage{amsfonts} 
\usepackage{mathtools}
\usepackage{soul}
\usepackage[dvipsnames]{xcolor} 
\usepackage{amsmath}
\usepackage{lipsum}
\usepackage{url}
\newtheorem{propo}{Proposition}
\newtheorem{definition}{Definition}
\newtheorem{remark}{Remark}
\newtheorem{example}{Example}

\begin{document}

\title{Explicit formulas for adiabatic elimination with fast unitary dynamics}

\author{Angela Riva, Alain Sarlette, and Pierre Rouchon%
\thanks{The authors are at Laboratoire de Physique de l'Ecole normale supérieure, ENS-PSL, CNRS, Inria, Centre Automatique et Systèmes (CAS), Mines Paris, Université PSL, Sorbonne Université, Université Paris Cité, Paris, France. AS is also at Dept. of Electronics and Information Science, UGent, Belgium.}
}

\begin{abstract}
The so-called ``adiabatic elimination'' of fast decaying degrees of freedom in open quantum systems can be performed with a series expansion in the timescale separation. The associated computations are significantly more difficult when the remaining degrees of freedom (center manifold) follow fast unitary dynamics instead of just being slow. This paper highlights how a formulation with Sylvester's equation and with adjoint dynamics leads to systematic, explicit expressions at high orders for settings of physical  interest.
\end{abstract}

\maketitle

\section{Introduction}

Using quantum effects in practical technological platforms relies on balancing two seemingly contradictory requirements: maintaining isolation from the environment to protect the fragile quantum coherences, and simultaneously allowing for fast manipulation towards performing operations. This is the fundamental problem that quantum control seeks to solve. 
The specific setups engineered in this context are often composed of several components, whose dimensions multiply and span several timescales. The ``useful quantum information'' represents only a small, ideally non-decaying part of this system. Fast decaying degrees of freedom (DOFs) are thus associated to auxiliary control signals, as e.g.~when stabilizing a quantum system via dissipation engineering \cite{poyatos1996quantum}, \cite{metelmann2015nonreciprocal}, \cite{mirrahimi2014dynamically}. Model reduction techniques seek to summarize the effect of these fast stabilizing DOFs on the dynamics of the target variables, preferably with operational principles and analytic expressions to guide system design. 

The linearity of the master equation for Markovian open quantum systems allows, in principle, a spectral decomposition very much like Hamiltonian (block-)diagonalization. Specific approaches in the physics literature rather focus on interpreting dissipative systems as being subject to uncertain events \cite{albash2012quantum, finkelstein2020adiabatic, reiter2012effective}. While this may be natural in many situations, it fails to provide a general and systematic approach. Our research group has initiated a  more mathematical line of work: identifying, via series expansion, the variables and dynamics corresponding to an invariant subspace where dynamics is at the slowest timescale. The approach was applied to perform model reduction inside a single system (Cartesian product, \cite{azouit2016adiabatic}) and in composite systems (tensor product, \cite{azouit2017towards}). It has recently been translated to the Heisenberg picture, with particular benefits for multi-partite systems \cite{RegentRouchonCDC2023,RegentRouchonPRA2024}.

All previously mentioned studies address scenarios where a rapid relaxation, characterized by eigenvalues with negative real parts, surrounds a slower evolution on an invariant subspace associated with eigenvalues near zero. However, like in the center manifold theory, the same model reduction technique should also admit a rapid relaxation towards an invariant subspace \emph{featuring large, almost purely imaginary eigenvalues} (unitary dynamics). This would fit practical situations where e.g.~the stabilized state keeps turning at significant characteristic frequencies, prefiguring in/out of resonance control \cite{SamPap}. Such a case was first treated in \cite{forni2018adiabatic,forni2019palette}, establishing the abstract expressions up to second-order expansion, under technical conditions on the fast unitary dynamics. In the present work, we provide integral formulas removing these conditions (Section \ref{sec:result}). We then highlight how, for typical systems of interest --- namely, eliminating a rapidly decaying harmonic oscillator environment --- these integral formulas also lead more easily to interpretable explicit expressions for the long-term dynamics on the invariant manifold (see the physical examples of section \ref{sec:examples}). The two key technical ingredients are the use of stable Sylvester equation solutions (Proposition \ref{prop1}) and the treatment of the typical dynamics on the adjoint state (Proposition \ref{prop2}).

\section{Model and Methods}\label{sec:model}

This section recalls the setting and adiabatic elimination approach as used e.g.~in \cite{azouit2017towards,forni2018adiabatic,tokieda2022completev2,RegentRouchonCDC2023}.

We consider a density operator $\rho_t$ at time $t$, defined over a bipartite Hilbert space $\mathcal{H} = \mathcal{H}_A \otimes \mathcal{H}_B$. The time evolution of this operator is described by the linear dynamics:
\begin{equation}
    \label{eq:model}
    \dot \rho_t = \mathcal{L} (\rho_t) = -i (H_A \otimes {I}_B)^\times(\rho_t)
    + \mathcal{I}_A \otimes \mathcal{L}_B(\rho_t) 
    -i g [H_{I}, \rho_t]. 
\end{equation}
Here $\mathbb \mathbb{I}_M$ is the identity operator and $\mathcal{I}_M$ the identity superoperator, associated to Hilbert space $\mathcal{H}_M$. We denote $H^\times(\bullet) \coloneqq [H, \bullet]$ the commutator with the Hamiltonian hermitian operator $H$, i.e.~$[H,X]=HX-XH$. The Lindbladian term for mode $B$ takes the general form:
\begin{equation}\label{eq:fast_lindbladian_model}
\begin{aligned}
    \mathcal{L}_B(\bullet) = -i H_B^\times(\bullet) + \kappa\; {\textstyle \sum_k}\, \mathcal{D}[L_k](\bullet)\; ,
\end{aligned}
\end{equation}
with $\kappa$ the damping rate and $\mathcal{D}[L](\bullet) \coloneqq L \bullet L^\dagger -\frac{1}{2} (L^\dagger L \bullet + \bullet L^\dagger L)$  the dissipation channel associated to operator $L$. The coupling  between $\mathcal{H}_A$ and  $\mathcal{H}_B$  is weak and Hamiltonian: 
$$g H_I = g\;  {\textstyle \sum_{k}} A_k \otimes B_k$$
with $A_k$ and $B_k$  operators on $\mathcal{H}_A$ and  $\mathcal{H}_B$ respectively and $g$ a small positive parameter.  

The adiabatic elimination procedure is based on the assumption that the system exhibits a clear separation of timescales, allowing for a perturbative expansion in the small parameter $\epsilon=g/\kappa \ll 1$. The dynamics for $g=\epsilon=0$ is assumed ``relatively easy to compute'', with $\mathcal{L}_B$ converging exponentially towards an equilibrium state $\bar \rho_B$ and subsystem $A$ independently rotating with a Hamiltonian $H_A$; thus, $\mathcal{L}(\rho)$ then features an invariant subspace $\mathcal{S}_0$ spanned by $\{ \rho = \rho_s \otimes \bar\rho_B \text{, for any } \rho_s \text{ on } \mathcal{H}_A\}$ and on which all eigenvalues are purely imaginary (center manifold). For $\epsilon \neq 0$, the interaction Hamiltonian $H_I$ acts as a perturbation on this situation. The goal is then to efficiently compute the ``long-term dynamics'', taking place on $\mathcal{S}_\epsilon$, an invariant subspace $\epsilon$-close to $\mathcal{S}_0$ and characterized by the eigenoperators of $\mathcal{L}(\rho)$ associated to eigenvalues with real part close to zero.

When $\epsilon \neq 0$, $H_I$ creates correlations between subsystems $A$ and $B$. As a result, the invariant subspace $\mathcal{S}_\epsilon$ does not take the same form as $\mathcal{S}_0$ anymore. Still, it is convenient to parameterize states on $\mathcal{S}_\epsilon$ by using $\rho_s \otimes \bar\rho_B \in \mathcal{S}_0$ as coordinates. Indeed, the dynamics are then reduced to $\rho_s$ whose evolution can be interpreted in comparison to the case $\epsilon=0$. To characterize the long-term behavior of the system, we aim to find two linear, time independent maps: $\mathcal{L}_s$, describing the dynamics of the reduced system coordinates $\dot \rho_s = \mathcal{L}_s(\rho_s)$; and $\mathcal{K}$, mapping $\rho_s$ to the solution $\rho = \mathcal{K}(\rho_s) \in \mathcal{S}_\epsilon$ of the complete dynamics \eqref{eq:model}. Expressing that the evolution of $\rho_s$ must mirror the one of $\mathcal{K}(\rho_s)$ leads to the condition:
\begin{equation}\label{eq:invariance_general}
    \mathcal{K}(\mathcal{L}_s (\rho_s)) = \mathcal{L}(\mathcal{K}(\rho_s)).
\end{equation}

Following usual approximation theory, we solve \eqref{eq:invariance_general} by expanding the maps $\mathcal{L}_s$ and $\mathcal{K}$ in powers of $\epsilon \ll 1$,
\begin{equation}\label{eq:expansions}
    \mathcal{L}_s(\rho_s) = \sum_{j=0}^\infty \epsilon^j \mathcal{L}_{s,j}(\rho_s), \quad \mathcal{K}(\rho_s) = \sum_{j=0}^\infty \epsilon^j \mathcal{K}_{j}(\rho_s) \; ,
\end{equation}
imposing to satisfy \eqref{eq:invariance_general} separately at each order $\epsilon^j$. Since the lowest-order contributions are the dominant ones, usually $\mathcal{L}_{s,j}$, $\mathcal{K}_j$ are computed explicitly for the first few orders only. 
Note that these solutions are not unique. For instance, for any solution $\mathcal{L}_s,\mathcal{K}$ of \eqref{eq:invariance_general} and $U$ any fixed unitary, $\mathcal{L}'_s(\bullet) = U \mathcal{L}_s(U^\dagger \bullet U) U^\dagger$ and $\mathcal{K}'(\bullet) = \mathcal{K} (U^\dagger \bullet U)$ are also solution of \eqref{eq:invariance_general}. This effectively corresponds to different choices of mapping between $\rho_s$ and $\mathcal{S}_{\epsilon}$ (reviewable at each order in $\epsilon$).

The most natural solution to \eqref{eq:invariance_general} at order $\epsilon^0$ is
\begin{equation} \label{eq:order0}
    \mathcal{L}_{s,0}(\rho_s) = -i H_A^\times(\rho_s), \quad \mathcal{K}_{0}(\rho_s) = \rho_s\otimes \bar \rho_B \, .
\end{equation}
In this paper, we further specify the mapping choice by imposing the \textit{partial trace gauge}, namely $\rho_s = \Tr_B(\rho) = \Tr_B(\mathcal{K}(\rho_s))$ at all orders of approximation. This gauge choice implies that $\Tr_B\mathcal{K}_1 = \Tr_B\mathcal{K}_2 = ... \equiv 0$.
The partial trace $\Tr_B$ over $\mathcal{H}_B$ of the superoperator $\mathcal{K}$ is naturally defined by first applying $\mathcal{K}$ to any operator $X$ (extended from the positive and Hermitian $\rho_s$ by linearity), then taking the standard partial trace of the result.

Note that when $H_A=0$, we get $\mathcal{L}_{s,0}=0$ which simplifies all further orders (see e.g.~\cite{azouit2017towards}). Our contributions are meant to address the difficulties appearing in solving \eqref{eq:invariance_general},\eqref{eq:expansions} at orders $\epsilon^1, \, \epsilon^2$ for  $H_A\neq 0$. In contrast with \cite{forni2019palette,forni2018adiabatic}, our explicit results require no assumptions on $H_A$ and they remain practical when the dimension of $\mathcal{H}_A$ is infinite. The only important assumption is the fast convergence of $\mathcal{L}_B$.

\section{Result} \label{sec:result}

Obtaining according to \eqref{eq:invariance_general},\eqref{eq:expansions} the elements $\mathcal{L}_{s,0},\mathcal{K}_0$ (see \eqref{eq:order0}) and 
\begin{equation}\label{eq:ls1}
   \epsilon \mathcal{L}_{s,1}(\rho_s) = - i g \; {\sum_k} \, \Tr(B_k \bar\rho_B) [A_k, \rho_s]
\end{equation}
(see proof of Proposition \ref{prop1} below) poses no particular problem. Our results concern the next order. Proposition \ref{prop1} derives a new integral closed form for the map $\mathcal{K}_1$ by solving a Sylvester equation. In Proposition \ref{prop2}, the second-order reduced dynamics $\mathcal{L}_{s,2}$, involving the commutator between $H_I$ and $\mathcal{K}_1$, is re-expressed with this form, putting the integral on a dual operator which is more systematically manageable in typical situations. Examples illustrate the usefulness of these formulas in Section \ref{sec:examples}. We start by defining some simplifying notation. 

\begin{definition}
    Operator $B_{0,k}$ on $\mathcal{H}_B$ is defined as
    \begin{equation}
    B_{0,k} = B_k - \Tr_B(B_k \bar \rho_B){I}_B.
    \end{equation}
\end{definition}
\begin{definition}\label{def:Ak}
    Operator $A_{k}^-(t)$ on $\mathcal{H}_A$ is defined as
    \begin{equation}
        A_k^- (t) =  e^{-i t H_A^\times}(A_k)= e^{-i t H_A}A_k e^{i t H_A}.
    \end{equation}
    (Note that this is the opposite propagation to the Heisenberg picture evolution under $H_A$, see Remark \ref{rem1}.)

\end{definition}

\begin{propo}\label{prop1}
Consider model \eqref{eq:model} and the adiabatic elimination expansion of Section \ref{sec:model} with partial trace gauge. The first order map $\mathcal{K}_1$ is given by
    \begin{equation}
    \begin{aligned}\label{eq:k1}
        \epsilon \mathcal{K}_1(\rho_s) = - i g   \sum_k \int_0^{+\infty} \Big\{ A_k^-(t) \rho_s \otimes e^{t \mathcal{L}_B}(B_{0,k} \bar \rho_B) 
         - \rho_s A_k^-(t) \otimes e^{t \mathcal{L}_B}(\bar \rho_B B_{0,k} )\Big\} dt.
    \end{aligned}
    \end{equation}
\end{propo}

\begin{proof}
    By inserting \eqref{eq:expansions} in the general invariance equation \eqref{eq:invariance_general} and grouping terms of order $\epsilon$, we obtain:
    \begin{equation}\label{eq:invariance_I}
         -i [H_A\otimes {I}_B, \epsilon\mathcal{K}_1(\rho_s)] + \mathcal{I}_A \otimes \mathcal{L}_B (\epsilon\mathcal{K}_1 (\rho_s)) \\ - i g [H_{I}, \mathcal{K}_0(\rho_s)]
         = \mathcal{K}_0(\epsilon\mathcal{L}_{s,1}(\rho_s)) + \epsilon\mathcal{K}_1(\mathcal{L}_{s,0}(\rho_s)).
    \end{equation}
    This equation has two unknowns: $\mathcal{L}_{s1}$ and $\mathcal{K}_1$. Two terms involving $\mathcal{K}_1$ are due to $H_A \neq 0$ and therefore new with respect to \cite{azouit2017towards}. Nevertheless, thanks to the gauge choice with $\Tr_B(\mathcal{K}_1)=0$, we can much like in \cite{azouit2017towards} apply $\Tr_B$ to \eqref{eq:invariance_I} and directly isolate $\mathcal{L}_{s1}$ in the form \eqref{eq:ls1}. To get this, we further observe that $\Tr_B(\mathcal{L}_B)=0$ and $\Tr_B(H_A^\times(\bullet))=H_A^\times(\Tr_B(\bullet))$. Next, inserting \eqref{eq:ls1} into \eqref{eq:invariance_I} yields the condition:
    \begin{equation}\label{eq:Sylvester}
            -i[H_A \otimes {I}_B, \epsilon\mathcal{K}_1(\rho_s)] + \mathcal{I}_A \otimes \mathcal{L}_B (\epsilon\mathcal{K}_1(\rho_s)) + i \epsilon\mathcal{K}_1  \left([H_A, \rho_s]\right) 
            = ig\, {\textstyle \sum_k}\, \left[A_k \otimes B_{0,k}, \rho_s \otimes \bar \rho_B \right],
    \end{equation}
    to be solved for $\mathcal{K}_1$. At the level of superoperators, \eqref{eq:Sylvester} is in fact a Sylvester equation of the form $AX + XB = C$, where $\epsilon\mathcal{K}_1$ plays the role of the unknown term $X$. To be explicit: $A \to -i (H_A\otimes{I}_B)^\times(\bullet) + \mathcal{I}_A \otimes \mathcal{L}_B (\bullet)$, $B \to i H_A^\times(\bullet)$, $C \to i g \sum_k [A_k \otimes B_{0,k}, \bullet \otimes \bar \rho_B]$. The solution to this Sylvester equation can be written as
    \begin{equation}\label{eq:SylSol}
            X = - \int_0^{+\infty} e^{tA} C e^{tB} dt,
    \end{equation}
    provided 
    \begin{equation}\label{eq:SylvCondition}
        \lim_{t\to + \infty} e^{tA} C e^{tB} = 0 \; .
    \end{equation}
    Using that (super)operators acting on different subsystems commute, as well as $e^{iH^\times}(Q) = e^{iH} Q e^{-iH}$, we obtain:
    \begin{equation}\nonumber
        \begin{aligned}
            &e^{tA} C e^{tB} (\rho_s) \, / g \\
             &= i \sum_k e^{t((-iH_A\otimes\mathbb{I}_B)^\times+\mathcal{I}_A\otimes\mathcal{L}_B)} \big[A_k \otimes B_{0,k}, e^{it H_A^\times}(\rho_s) \otimes \bar \rho_B\big] \\
            &= i \sum_k \Big( e^{-it H_A^\times}(A_k) \rho_s \otimes e^{t \mathcal{L}_B}\big(B_{0,k} \bar \rho_B\big) - \rho_s e^{-it H_A^\times}(A_k) \otimes e^{t \mathcal{L}_B}\big( \bar \rho_B B_{0,k}\big) \Big) \; .
        \end{aligned}
    \end{equation}
This last expression corresponds to the proposition  statement. Since $\lim_{t\to +\infty} e^{t \mathcal{L}_B}(\bullet) = \Tr(\bullet) \bar\rho_B$, we have $\lim_{t \to +\infty} e^{t \mathcal{L}_B}(B_{0,k} \bar \rho_B)  = \lim_{t \to +\infty} e^{t \mathcal{L}_B}(\bar \rho_B B_{0,k}) = 0$, such that condition \eqref{eq:SylvCondition} is satisfied. 
\end{proof}

\begin{remark}\label{rem1}
Note that $A_k^-(t)$ in the statement of Proposition \ref{prop1} follows the opposite dynamics to the Heisenberg evolution under $H_A$. We can understand this as follows. According to $\mathcal{L}_{s,0}$, $\rho_s$ already rotates with $H_A$. Inside the integral of \eqref{eq:k1}, we must replace it by $(A_k \rho_s)$ rotating with $H_A$, much like the second tensor factor takes $(B_{0,k} \bar{\rho}_B)$ evolving under $\mathcal{L}_B$.
\end{remark}

Compared to previous work \cite{forni2018adiabatic,forni2019palette}, The proposition \ref{prop1} avoids any technical conditions and the need to solve several equations by treating the components of $H_A$ individually. In turn, it does not guarantee positivity --- whose importance has anyways been re-evaluated since \cite{tokieda2022completev2} --- and it leaves a propagator equation to solve, for all $t$ and separately for each subsystem. Propagation on subsystem A, with a fixed Hamiltonian $H_A$, is usually no big deal. Often, $H_A$ can be diagonalized efficiently or commuted through typical operators to compute $A_k^-(t)$.  Propagation on subsystem B, with a Lindbladian superoperator, is computationally less convenient. Indeed, even for simple steady-state cases like convergence towards a thermally broadened vacuum, it is not trivial to express the state at any time along the trajectory starting from rather arbitrary initial states like $B_{0,k} \bar\rho_B$. The following result mitigates this difficulty.

\begin{definition}\label{def:Bk}
    The operator $B_{k}(t)$ on $\mathcal{H}_B$, evolving under the action of $\mathcal{L}_B$ in the Heisenberg picture, is defined as:
    \begin{equation}\label{eq:Bkt}
        B_k (t) =  e^{t \mathcal{L}_B^*}(B_k),
    \end{equation}
    with the adjoint $\mathcal{L}^*$ of the Lindbladian  defined as
    $\mathcal{L}^*(\bullet) = +i [H, \bullet] + \kappa \sum_k \left( L_k^\dagger \bullet L_k -\frac{1}{2} (L_k^\dagger L_k \bullet + \bullet L_k^\dagger L_k) \right)$.
\end{definition}

\begin{propo}\label{prop2}
Consider model \eqref{eq:model} and the adiabatic elimination expansion of Section \ref{sec:model} with partial trace gauge. The second order  reduced dynamics is given by:
    \begin{equation}\label{eq:prop}
    \begin{aligned}
        \epsilon^2\, \mathcal{L}_{s,2} = -\frac{g^2}{\kappa}& \sum_{k,l}  \int_0^{+\infty}  \Big\{   \Tr\{B_l (t) B_{0,k} \bar\rho_B\} \big[A_l,\, A_k^-(t)\rho_s \big] - \Tr\{B_l (t)\bar\rho_B B_{0,k}\} \big[A_l,\, \rho_s A_k^-(t)  \big]\Big\}\,dt.
    \end{aligned}
    \end{equation}
\end{propo}

\begin{proof}
In the second order invariance condition: 
{\small 
     \begin{equation}
        \begin{aligned}
            -i [H_A \otimes I_B, \epsilon^2\mathcal{K}_2(\rho_s)] + \mathcal{I}_A \otimes \mathcal{L}_B (\epsilon^2\mathcal{K}_2 (\rho_s)) - i [gH_{I}, \epsilon\mathcal{K}_1(\rho_s)] 
         = \mathcal{K}_0(\epsilon^2\mathcal{L}_{s,2}(\rho_s))  + \epsilon^2 \mathcal{K}_1(\mathcal{L}_{s,1}(\rho_s)) + \epsilon^2\mathcal{K}_2(\mathcal{L}_{s,0}(\rho_s)),
        \end{aligned}
    \end{equation}}
take the partial trace over subsystem B, like we did before Proposition \ref{prop1} towards getting $\mathcal{L}_{s,1}$, in order to obtain:
    \begin{equation*} 
        \epsilon^2 \mathcal{L}_{s,2}(\rho_s) = -i \, \Tr_B\{[gH_{I}, \epsilon \mathcal{K}_1 (\rho_s)]\}.
    \end{equation*}
Substituting $H_I$, and $\mathcal{K}_1$ from Proposition \ref{prop1}, yields:
    \begin{equation}\label{eq:trace}
    \begin{aligned}
        \epsilon^2 \mathcal{L}_{s,2}(\rho_s) = 
        (-ig)^2 \sum_{l,k}\Big\{ \left[A_l,\, \int_0^{+\infty} c_{k,l}(t) A_k^-(t)\, dt\, \rho_s \right]
        - \left[ A_l, \, \rho_s  \int_0^{+\infty} \tilde{c}_{k,l}(t) A_k^-(t)\, dt \right] \Big\}, 
    \end{aligned}
    \end{equation}
where {\small
    \begin{equation}\label{eq:heis}
    \begin{aligned}
        & c_{k,l}(t) = \Tr(B_l e^{t \mathcal{L}_B}(B_{0,k} \bar\rho_B)) = \Tr(e^{t \mathcal{L}_B^*}(B_l) B_{0,k} \bar\rho_B)\\
        & \tilde{c}_{k,l}(t) = \Tr(B_l e^{t \mathcal{L}_B}( \bar\rho_B B_{0,k})) = \Tr(e^{t \mathcal{L}_B^*}(B_l)\bar\rho_B B_{0,k}).
    \end{aligned}
    \end{equation}
    }
On the right-hand side of \eqref{eq:heis}, we have transferred the Lindblad dynamics to the adjoint, in other words evolving the partner operator inside the trace in Heisenberg picture. Recalling Definition \ref{def:Bk}, this corresponds to the statement.  \hfill $\square$
\end{proof}

\begin{remark}
The Proposition does not claim that the resulting reduced dynamics preserves positivity --- thus taking the typical Lindblad form with positive dissipation rates. This property is proven in \cite{azouit2017towards} for $H_A=0$, and in \cite{forni2019palette} under more technical conditions.
\end{remark}

The essential part of Proposition \ref{prop2} is to replace, when computing $\mathcal{L}_{s,2}$, the Lindbladian trajectory starting at $(B_{0,k} \bar{\rho}_B)$ by an adjoint Lindbladian trajectory on $B_l$. As already noted in \cite{Jeremie}, the latter is often much easier to compute in typical situations. For instance, the operator propagation can be computed with little more effort than a classical system, when both $\mathcal{L}_B$ and $B_l$ on subsystem B correspond to a so-called linear quantum system \cite{nurdin2017linear}.

The strategy followed by these two Propositions addresses the same major issue in resolving the series expansion \eqref{eq:invariance_general},\eqref{eq:expansions}, namely the inversion (in matrix inverse sense) of the fast dynamics. When pursuing the expansion at higher orders, the main computational issue remains the same, and we can reiterate the same procedure to obtain explicit expressions.

\begin{propo}\label{prop3}
    Consider model \eqref{eq:model} and the adiabatic elimination expansion of Section \ref{sec:model} with partial trace gauge. Assume that the terms of the series expansion have been computed up to $n-1$. Then, for all $n\geq 1$,
    \begin{equation}\label{eq:Lsn}
        \epsilon \mathcal{L}_{s,n}(\rho_s) = -ig \sum_k \big[A_k \, , \, \Tr_B\big( (I_A \otimes B_k) \mathcal{K}_{n-1}(\rho_s) \big) \big]
    \end{equation}
    and $\epsilon \mathcal{K}_n$ can be computed with an integral formula similar to \eqref{eq:k1}.
\end{propo}

\begin{proof}
    The invariance condition at order $\epsilon^n$ writes:
    \begin{multline} \label{eq:high_order}
    \epsilon \mathcal{K}_0(\mathcal{L}_{s,n}(\rho_s)) =
    \epsilon \left(-i(H_A\otimes I_B)^\times + \mathcal{I}_A \otimes \mathcal{L}_B \right)\left( \mathcal{K}_n (\rho_s)\right) 
    - \epsilon
    \mathcal{K}_n\left( \mathcal{L}_{s,0}(\rho_s) \right) \\
    - \epsilon \sum_{m=1}^{n-1} \mathcal{K}_m\left( \mathcal{L}_{s,n-m}(\rho_s) \right)
    -ig (H_I^\times)(\mathcal{K}_{n-1}(\rho_s)) \; .
\end{multline}
The last line contains only known terms at stage $n$. Like in Proposition \ref{prop1}, by taking the partial trace over subsystem B, all the terms involving $\mathcal{K}_n$ vanish with the gauge choice $\rho_s=\Tr_B(\rho)$, while $\Tr_B \mathcal{K}_0 = \mathcal{I}_s$. The terms in $\mathcal{K}_m$ also vanish by partial trace. We thereby obtain the explicit expression \eqref{eq:Lsn} for $\mathcal{L}_{s,n}$, assuming all previous orders were known.

After this, the left-hand side of \eqref{eq:high_order} is also known and the equation takes the Sylvester form $AX+XB=C$ like in the proof of Prop.\ref{prop1}, with $A$ and $B$ unchanged, $\epsilon \mathcal{K}_n$ playing the role of $X$, and $C$ containing all known terms. There remains to show that the integral form solution \eqref{eq:SylSol} converges at this order. Like in the proof of Prop.\ref{prop1}, this is ensured thanks to $\lim_{t\to + \infty} e^{tA}(\bullet) = \mathcal{U}_{t,A}(\Tr_B(\bullet)) \otimes \bar\rho_B$ with $\mathcal{U}_{t,A}$ a unitary evolution, and the annihilation of terms under $\Tr_B$. 
\end{proof} 

When one is only interested in the dynamics $\mathcal{L}_{s,n}$, it may be possible to also apply the trick of Proposition \ref{prop2} iteratively at all orders. This would simplify computations significantly, compared to computing each $\mathcal{K}_m$ explicitly. Investigating this possibility is left for future work.


\section{Examples} \label{sec:examples}

In this section, we give some practical demonstrations of the results presented in Section \ref{sec:result}. The starting point is a harmonic oscillator B subsystem undergoing so-called quantum linear dynamics \cite{nurdin2017linear}:
\begin{equation}\label{eq:fast_lindbladian}
\begin{aligned}
    \mathcal{L}_B(\bullet) = -i \omega_B[b^\dagger b, \bullet] + \kappa_\phi \mathcal{D}[b^\dagger b](\bullet)  + \kappa(1 + n_\text{th})\mathcal{D}[b](\bullet)+\kappa n_\text{th} \mathcal{D}[b^\dagger](\bullet) ,
\end{aligned}
\end{equation}
with $b$ ($b^\dagger$) the bosonic annihilation (creation) operator, $b^\dagger b$ the number operator, $\omega_B$ the undamped harmonic oscillator frequency, $\kappa$ the damping rate for the dissipator $\mathcal{D}$, and $n_\text{th}$ the residual thermal excitation. 
We note that the unique steady state such that $\mathcal{L}_B (\rho) = 0$ is the thermal state,
\begin{equation}\label{eq:thermal}
    \bar \rho_B = \sum_{n=0}^{+\infty} \frac{n_\text{th}^n}{(n_\text{th}+1)^{n+1}}\ketbra{n}
\end{equation}
in the number-basis $\ket{n}$, i.e.~where $b^\dagger b = \sum_{n=0}^{+\infty} n \ketbra{n}{n}$.

We couple this system to an arbitrary A subsystem, using a standard dipolar coupling which is also linear:
\begin{equation}\label{eq:HIbare}
\tilde{H}_I = \tilde{A}\otimes (b + b^\dagger) \; .
\end{equation}

In the vast majority of cases \cite{Haroche}, the timescale separation between $\omega_B$ and $g$ is so large that significant coupling effects on A only happen when $H_A$ contains frequencies close to $\omega_B$. This is what lies behind resonance conditions invoked in more high-level descriptions of physical phenomena. More precisely, going to a rotating frame with $\omega_B$ and averaging the resulting time-dependent dynamics singles out the relevant dominant long-term behavior \cite{sanders2007averaging}. Note that this is mathematically justified only if $\omega_B \gg \kappa$ as well. In physics terms \cite{Haroche}, the resonance has a width of order $\kappa$, thus frequency selection only works for $\kappa \ll \omega_B$.

With the framework developed here, we are able to compare the results of adiabatic elimination in both cases:
\newline [Section \ref{ssec:JaynesCummings}] The physics traditional way, i.e.~first averaging out the large part of $H_A$ in a rotating frame; 
\newline [Section \ref{ssec:Dipolar}] Maintaining all terms in so-called ``inertial frame'', with thus $\omega_B$ and $H_A$ of the same order and non-negligible compared to $\kappa$. The results should match those of Section \ref{ssec:JaynesCummings} when $\omega_B$ and $H_A$ are close and large compared to $\kappa$; otherwise they are new.

\subsection{Jaynes-Cummings interaction}\label{ssec:JaynesCummings}

The Jaynes-Cummings interaction, see e.g.~\cite[Chap.3.4]{Haroche}, corresponds to the canonical theoretical and practical case of coupling subsystem B with a (quasi-) resonant qubit A. After first-order averaging in a frame where both subsystems rotate at the resonance frequency, the interaction $H_I$ boils down to an energy exchange term.

\begin{example}
Consider the bipartite system composed of a harmonic oscillator and a two-level system (qubit), whose dynamics is described by 
    \begin{eqnarray}\label{eq:fastJC}
    \mathcal{L}_B(\bullet) &=& -i \Delta[b^\dagger b, \bullet] + \kappa(1 + n_\text{th})\mathcal{D}[b](\bullet) +\kappa n_\text{th} \mathcal{D}[b^\dagger](\bullet) + \kappa_\phi \mathcal{D}[b^\dagger b](\bullet)\; \nonumber ,\\ \nonumber
    H_A&=&0 \; ,\\ \nonumber
    H_I &=& \sigma_+ \otimes b + \sigma_- \otimes b^\dagger,
\end{eqnarray}
with $\mathcal{H}_A = \text{span}\{\ket{g},\ket{e}\}$ and $\sigma_- = (\sigma_+)^\dagger = \ketbra{g}{e}$ the qubit lowering operator (from excited to ground state). With Proposition \ref{prop2}, assuming $\kappa \gg g$, the second order reduced dynamics is given by
\begin{multline}\label{eq:JCresult}
        \frac{d}{dt} \rho_s = \epsilon^2\mathcal{L}_{s,2}(\rho_s) = -i (1+2n_{th})\frac{4 \Delta g^2}{\abs{\gamma}^2}\left[\frac{\sigma_z}{2}, \rho_s \right] 
        + (1+n_\text{th}) (\kappa + \kappa_\phi)\frac{4 g^2}{\abs{\gamma}^2} \mathcal{D}[\sigma_-](\rho_s)  
         + n_\text{th}(\kappa + \kappa_\phi)  \frac{4 g^2}{\abs{\gamma}^2} \mathcal{D}[\sigma_+](\rho_s),\;
\end{multline}
where $\sigma_z = \ketbra{g}{g}-\ketbra{e}{e}$ and $\gamma = \kappa + \kappa_\phi + 2 i \Delta$.
\end{example}
\noindent\textit{Solution.}
The coupling corresponds to our general setting with $B_1 = b$, $B_2 = b^\dagger$, $A_1 = \sigma_+$, $A_2 = \sigma_-$. The zeroth-order dynamics $\mathcal{L}_{s,0}$ vanishes since $H_A=0$. Note that $\mathcal{L}_B$ in \eqref{eq:fastJC} still features the unique steady state \eqref{eq:thermal}, and since $\bra{n}b\ket{n}=0$ for all $n$ we also have $\mathcal{L}_{s,1}=0$.
    To recover the second order reduced dynamics, we start by computing the Heisenberg representation of the coupling operators (Definition 3):
    \begin{equation}\label{eq:time-dep-b}
    \begin{gathered}
        \frac{d}{dt}b = \mathcal{L}_B^*(b) = - \tfrac{\gamma}{2} b, \\ \text{thus} \quad b (t) = e^{- \tfrac{\gamma}{2} t} b(0),
    \end{gathered}
    \end{equation}
with $\gamma$ defined as in the statement. In this step, we highly benefit from assuming a  linear quantum system \cite{nurdin2017linear} on B in order to obtain such easy closed-form solution. Indeed, for a general quantum system B, the computation of \eqref{eq:Bkt} in Definition 3 can be a significant bottleneck.

From there, with easy calculations involving the geometric series, we obtain the coefficients $c_{k,l}(t)$ and $\bar c_{k,l}(t)$ mentioned in \eqref{eq:heis}. Since $H_A=0$, the operators $A_k^-(t)$ are time independent, so the integral to be solved in Proposition \ref{prop2} is simply $\int_0^{+\infty} e^{- s \gamma} ds = \tfrac{1}{\gamma}$. By rearranging the terms, $\mathcal{L}_{s,2}$ can be put in standard Lindblad form to obtain \eqref{eq:JCresult}.
{\hfill$\square$}

The reduced model \eqref{eq:JCresult} thus contains a unitary shift with B pulling on the frequency of A (Hamiltonian in $\sigma_z$); and it translates thermal dissipations on B into corresponding dissipations in $\sigma_-$ and $\sigma_+$ on A, yet with decreasing effect as $\kappa$ increases. This result is well-known \cite{azouit2017towards}. It readily generalizes to a subsystem A of higher dimension. One just replaces, in the model and in the resulting dissipator, $\sigma_-$ by $A$ and $\sigma_+$ by $A^\dagger$; in the resulting slow Hamiltonian, one replaces $(1+2n_{th}) \tfrac{\sigma_z}{2}$ by $(n_{th} A A^\dagger - (1+n_{th})A^\dagger A)$.

\subsection{Treatment with fast unitary dynamics in inertial frame}\label{ssec:Dipolar}

We now consider the same setting \emph{in inertial frame}, i.e.~without going to rotating frame and averaging, but thus with Hamiltonian $\tilde{H}_A$ (and $\omega_B\, b^\dag b$) not small compared to $\kappa$. The model corresponds to \eqref{eq:fast_lindbladian}, \eqref{eq:HIbare} with, for a qubit, $\tilde{A} = (\sigma_-+\sigma_+) = \sigma_x$ and $H_A$ proportional to $\sigma_z$.
\begin{example}
    Consider the bipartite system described by:
    \begin{eqnarray}\label{eq:Ex2Setting}
        \mathcal{L}_B(\bullet) &=& -i \omega_B[b^\dagger b, \bullet] + \kappa(1 + n_\text{th})\mathcal{D}[b](\bullet) +\kappa n_\text{th} \mathcal{D}[b^\dagger](\bullet) + \kappa_\phi \mathcal{D}[n](\bullet)\; \nonumber ,\\ \nonumber
        H_A &=& -\omega_{\rm eg}\tfrac{\sigma_z}{2}, \nonumber \\
        H_I  &=& \sigma_x \otimes(b+b^\dagger), \nonumber
    \end{eqnarray}
    with the energy gap $\omega_{\rm eg} = \omega_e - \omega_g$, $\sigma_x = \ketbra{e}{g} + \ketbra{g}{e}$, $\sigma_z = \ketbra{g}{g}-\ketbra{e}{e}$. With Proposition \ref{prop2}, assuming $\kappa \gg g$, the second order reduced dynamics is given by  
\begin{multline}\label{eq:second_fast}
\tfrac{d}{dt}\rho_s = \mathcal{L}_{s,0}(\rho_s) + \epsilon^2 \mathcal{L}_{s,2}(\rho_s) =
  -i[-\omega_{\rm eg}\tfrac{\sigma_z}{2},\rho_s] -i g^2 Y \left[\frac{\sigma_z}{2}, \rho_s \right]
 + g^2 \sum_{\ell,\ell' \in \{+,-\}}  
X_{\ell \ell'} \left(\sigma_{\ell'} \rho_s \sigma_\ell^\dagger -\frac{\rho_s \sigma_\ell^\dagger \sigma_{\ell'} + \sigma_\ell^\dagger \sigma_{\ell'} \rho_s}{2}  \right)
\; , 
\end{multline}
with the hermitian matrix $X$ and coefficient $Y$ defined by: 
\begin{equation}
\begin{split}
    X_{\ell\ell'} &= r_{\ell'} + r_\ell^* + e_{\ell'} + e_\ell^* \\
    Y &= \tfrac{1}{2i}(
    r_++e_+-r_+^*-e_+^*-r_--e_-+r_-^*+e_-^*),
\end{split}
\end{equation}
and the coefficients $r_{\ell}$ and $e_{\ell}$, $\ell \in \{+, -\}$: 
\begin{equation}
        r_\pm  = \frac{2(1+n_{\rm th})}{\gamma_{\pm}}, \quad e_\pm = \frac{2n_{\rm th}}{{\gamma_{\mp}}^*},
\end{equation}
where $\gamma_{\pm} = \kappa + \kappa_\phi + 2 i (\omega_B \pm \omega_{\rm eg})$.
\end{example} 

\noindent\textit{Solution.}
The coupling corresponds to the above setting with just $A_1 = \sigma_x$, $B_1 = b + b^\dagger$. The zero order, given by \eqref{eq:order0}, is $\mathcal{L}_{s,0} = -i \omega_{eg} \sigma_z^\times$. The first order $\mathcal{L}_{s,1}=0$ for the same reason as in the previous example. To obtain the second order reduced dynamics via Definition \ref{def:Bk}, the time dependence for $b(t)$ is the same as in \eqref{eq:time-dep-b}, with $\omega_B$ replacing $\Delta$ in the definition of $\gamma$. With this one can easily obtain according to \eqref{eq:heis} the sole coefficients $c_{11}(t) = (n_{th}+1)e^{-\gamma t/2}+ n_{th} e^{-\gamma^* t/2}$ and $\tilde c_{11}(t) = c_{11}(t)^*$. Since we just have $A_1^-(t) = \sigma_+ e^{-i \omega_{\rm eg}t} + \sigma_- e^{+i \omega_{\rm eg}t}$, we compute the time integrals 
\begin{equation}
\begin{split}
    &\int_0^{+\infty} c_{11}(t) A_1^-(t) dt = \sum_{\ell \in \{+, -\}} (r_\ell  + e_\ell ) \sigma_\ell, \\
    &\int_0^{+\infty} \tilde c_{11}(t) A_1^-(t) dt = \sum_{\ell \in \{+, -\}} (r_\ell^* + e_\ell^*) \sigma_\ell^\dagger.
\end{split}
\end{equation}
The statement then follows after straightforward calculations and relabeling.
{\hfill$\square$}\\

The result generalizes to higher-dimensional subsystems A without much difficulties. Replace $\sigma_x$ by $(A+A^\dagger)$ and $ H_A = \sum_{s=1}^{d_A} \omega_s \ketbra{s}$ in the problem statement. The computations of $c_{11}(t)$ and $\tilde{c}_{11}(t)$ involving exponentials remains unchanged. The time evolution of the operators $A_k^-(t)$ will be given by
\begin{equation}\label{eq:heisenberg_a}
    A_k^-(t) = \sum_{n,m=1}^{d_A} e^{-i(\omega_n - \omega_m)t} \bra{n}A_k\ket{m} \ketbra{n}{m} \; .
\end{equation}
With these expressions, we can thus in principle easily compute the integrals
\begin{equation}
    \int_{0}^{+\infty} c_{k,l}(t) A_k^-(t) dt \quad \text{and} \quad \int_0^{+\infty} \tilde{c}_{k,l}(t) A_k^-(t) dt,
\end{equation}
and obtain the second order reduced dynamics.\\

\noindent \textbf{Interpretation.} The following observations are in order about the result in \eqref{eq:second_fast}.
\begin{itemize}
\item One checks that for $\omega_{eg}=\omega_B$ large, averaging the reduced system in a frame rotating with $U(t) = e^{i \omega_{eg} t \sigma_z/2}$, yields back the result of Example 1 as expected. Without averaging, terms in $X_{+-}$ and $X_{-+}$ remain. Note though that the averaging condition now relaxes to $\omega_B,\omega_{eg} \gg \epsilon^2 \kappa$, thanks to confining ourselves to a manifold with slow dynamics.
\item For $\omega_{eg}=0$ instead, we have $r_{+}=r_-$ and $e_+=e_-$ such that $Y=0$ and $X$ is proportional to the all-ones matrix. This singularity implies a single dissipation channel, in $\sigma_x$ i.e.~proportional to the coupling operator, in agreement with the result of \cite{azouit2017towards} when $H_A=0$.
\item The formula \eqref{eq:second_fast} thus allows us to capture all intermediate scaling cases.
\item To have a completely positive Lindblad form, interpretable as a standalone open quantum system, the matrix $X$ in \eqref{eq:second_fast} should be positive. Here, $X$ has a positive trace, but its \emph{determinant is independent of $n_{th}$ and negative as soon as $\omega_{eg} \neq 0$}. Then  such interpretation fails, as also happened in \cite{tokieda2022completev2} for the partial trace gauge ($\rho_s = \Tr_B(\rho)$). This raises the natural question of whether an alternative gauge choice could restore a positive $X$, as seen in \cite{azouit2016adiabatic} for the second-order case with $H_A=0$, and in \cite{forni2019palette} for $H_A \neq 0$ under specific conditions. More recently, \cite{tokieda2022completev2} demonstrated that for $H_A=0$ at 4th order and for a class of parameters, no gauge choice can restore positivity. Further investigation into gauge choices and the preservation of positivity is left for future work. See also next items.
\item The corresponding Bloch equations for $\rho=\frac{I+x\sigma_x+y\sigma_y+z\sigma_z}{2}$ are:
\begin{eqnarray}
\nonumber 
\dot{x} &=& \left( \omega_{eg} -g^2 Y \tfrac{2n_{th}}{1+2n_{th}} \right)\, y \\
\nonumber 
\dot{y} &=& \left( -\omega_{eg} +g^2 Y \tfrac{2+2n_{th}}{1+2n_{th}} \right)\, x \;
- g^2(1+2n_{th})\, r_z\, y
\\
\label{eq:Bloch}
\dot{z} &=& - g^2(1+2n_{th}) r_z \, (z-\bar{z})
\end{eqnarray}
where $r_z = 4(\kappa+\kappa_\phi)(\frac{1}{|\gamma_+^2|}+\frac{1}{|\gamma_-^2|})$ and $\bar{z} = \frac{|\gamma_+|^2-|\gamma_-|^2}{(|\gamma_+|^2+|\gamma_-|^2)(1+2n_{th})}
$. For $\omega_{eg}\neq 0$ this system converges to $x=y=0$, $z=\bar{z}$. For $\omega_{eg}=0$, the $x$ coordinate remains invariant and the two others converge exponentially to $y=0$ and $z=0$. 
\item The equations \eqref{eq:Bloch} can thus best be seen as just coordinates, accurately describing the Lindblad form \eqref{eq:Ex2Setting} when restricted to an invariant subspace inside $\mathcal{H}_A \otimes \mathcal{H}_B$. In  \cite{tokieda2022completev2}, a necessary and sufficient inequality on the spectrum of Bloch equations is presented to decide whether another coordinate choice (thus not imposing $\rho_s = \text{Tr}_B(\rho)$) may yield a completely positive reduced model, identifiable with a qubit. Remarkably, for any parameter values, our particular result \eqref{eq:Bloch} appears to lie on the boundary of these inequalities. Higher orders of the expansion thus have to be examined before concluding.
\end{itemize}

\section{Conclusion}

This paper leverages the integral solution of Sylvester equation to propose explicit formulas for quantum model reduction via adiabatic elimination, when the remaining subsytem undergoes fast unitary dynamics. From a formal viewpoint, this completes the picture of spectral block-decomposition, by assuming only a timescale separation on the real part of the eigenvalues, i.e.~the degrees of freedom which vanish \emph{versus} remain in the long term. This is the linear version of a central manifold with non-trivial motion in general systems theory. From a practical viewpoint, we have illustrated how our results avoid the need to move to the interaction picture and perform an averaging approximation before treating the model reduction. 

Future work should allow us to prove if other gauge choices could restore complete positivity of the reduced model, like in  \cite{forni2019palette}, in our generalized setting too. Additionally, it should generalize the computations to any linear quantum system \cite{nurdin2017linear}, providing insights on how their various tunings can be integrated into quantum dissipation engineering. We would also use this tool to further explore corrections to averaging approximations, also called ``Rotating Wave Approximation'', RWA, and ubiquitous in quantum engineering by making things on/off-resonant. Indeed, while averaging expansions can in principle be carried out at higher order, they are a priori not converging, unlike the block-spectral decomposition of the present paper. 

\section*{Acknowledgment}

The authors would like to thank Mazyar Mirrahimi, Jérémie Guillaud, Samuel Deléglise, Masaaki Tokieda, Lev-Arcady Sellem and other colleagues for useful discussions. This research has been funded by ANR grants HAMROQS and MECAFLUX (French Research Agency), by Plan France 2030 through the project ANR-22-PETQ-0006, and by the European Research Council (ERC) under the European Union’s Horizon 2020 research and innovation program (grant agreement No. 884762).


\bibliographystyle{abbrv}
\bibliography{biblio}
%



\end{document}